\documentclass[11pt]{article}
\usepackage[margin=1in]{geometry}
\usepackage{amsmath,amssymb,amsthm}
\usepackage{mathtools,array}
\usepackage{hyperref}
\usepackage{xcolor}
\usepackage{microtype}

\hypersetup{
  colorlinks=true,
  linkcolor=blue!70!black,
  citecolor=green!50!black,
  urlcolor=blue!70!black
}

\newtheorem{theorem}{Theorem}[section]
\newtheorem{lemma}[theorem]{Lemma}

\newtheorem{proposition}[theorem]{Proposition}
\theoremstyle{definition}
\newtheorem{definition}[theorem]{Definition}

\newcommand{\E}{\mathbb{E}}
\newcommand{\Prb}{\mathbb{P}}
\newcommand{\R}{\mathbb{R}}
\newcommand{\I}{\mathbf{1}}
\newcommand{\cU}{\mathcal{U}}
\newcommand{\cG}{\mathcal{G}}
\newcommand{\cT}{\mathcal{T}}
\newcommand{\Unif}{\mathrm{Unif}}

\title{Envy-Free Allocations for Random Non-Additive Valuations}
\author{Pasin Manurangsi\\Google Research\\\texttt{pasin@google.com}}
\date{}

\begin{document}

\maketitle

\begin{abstract}
Benad\`e et al.~\cite{BHPV24} introduced a semi-random model for non-additive fair division in which each agent starts with an arbitrary base valuation over bundles of items and independently permutes the item names uniformly at random. While they proved that envy-free allocations exist with high probability as the number of items grows for $n = 2$ agents, they left open whether the same holds for $n > 2$ agents. We resolve this open question in the affirmative: for any fixed number of agents $n$ and any number of items $m$ divisible by $n$, an envy-free allocation exists with probability $1 - o(1)$ as $m \to \infty$.

Our proof combines the second moment method with a connection to the (small-set) expansion of generalized Johnson graphs, which may be of independent interest.
\end{abstract}

\section{Introduction}\label{sec:intro}

How to divide a collection of scarce resources fairly among multiple agents is a central problem at the intersection of economics, computer science, and multi-agent systems, with applications ranging from estate and divorce settlement to course allocation, rent division, and cloud computing~\cite{BT96,Moulin03,Budish11,BCM16}. A quintessential notion of fairness is \emph{envy-freeness (EF)}, which requires that every agent view her own share of the resources as at least as valuable as the share received by any other agent. While envy-free divisions are guaranteed to exist when resources are infinitely divisible (i.e., \emph{cake-cutting}), many practical domains involve \emph{indivisible} items that must be assigned in their entirety~\cite{Markakis17,Aziz20,Walsh20}.

When items are indivisible, envy-free allocations need not exist in the worst case: for instance, if a single desirable item is divided between two agents, the agent who does not receive the item will inevitably envy the other. However, such worst-case obstructions typically arise when there are very few items or when agents' preferences are adversarial and tightly aligned. In practice, agents' valuations often exhibit natural variability, and the number of items is frequently much larger than the number of agents. This observation motivated a rich line of research on \emph{asymptotic (or stochastic) fair division}, initiated by Dickerson et al.~\cite{DGK14}, which studies when EF allocations exist with high probability under natural probabilistic models of agent preferences as the number of items grows.

Until recently, nearly all work on asymptotic fair division focused on \emph{additive} valuations, where each agent has an independently drawn value for each item and the value of a set of items is simply the sum of the values of its elements~\cite{DGK14,KPW16,Suk16,MS17,MS20,MS21,BG22,BFGP22,MS25a,MSY25,GIMS25}. Extending asymptotic fair division to general, \emph{non-additive} valuations (which capture complementarities and substitutabilities among items) requires care even to formulate a meaningful stochastic model.
Despite this challenge, Benad\`e et al.~\cite{BHPV24} successfully formulated the \emph{random renaming} model, which is very general yet offers a rich collection of interesting (non-)existence results.

\paragraph{Random Renaming Model.} To discuss further, we first introduce some notation, starting with the basics of fair division and EF.
We consider the problem of dividing a set $[m] := \{1, \dots, m\}$ of $m$ indivisible items among a set $[n] := \{1, \dots, n\}$ of $n$ agents. A \emph{bundle} is any subset of items $S \subseteq [m]$. Each agent $i \in [n]$ has a \emph{valuation function} $u_i : 2^{[m]} \to \R$ that assigns a real value $u_i(S)$ to every bundle $S \subseteq [m]$. An \emph{allocation} is an ordered partition $A = (A_1, \dots, A_n)$ of $[m]$ into $n$ pairwise disjoint bundles, where $A_i \subseteq [m]$ is the bundle assigned to agent $i$. An allocation $A = (A_1, \dots, A_n)$ is \emph{envy-free (EF)} with respect to valuations $(u_1, \dots, u_n)$ if every agent $i \in [n]$ weakly prefers her own bundle $A_i$ to any other bundle $A_j$. That is, $u_i(A_i) \ge u_i(A_j)$ for all $i, j \in [n].$

To study non-additive valuations in a stochastic setting without imposing structural restrictions, Benad\`e et al.~\cite{BHPV24} introduced the \emph{random renaming} model. In this model, each agent $i \in [n]$ starts with an arbitrary\footnote{The valuation $v_i$ is not assumed to be monotone or subadditive. Thus, this model can capture goods, chores, and mixed manna with arbitrary complementarities or substitutabilities.} base valuation function $v_i : 2^{[m]} \to \R$ and draws an independent permutation $\pi_i \sim \Unif(S_m)$ of the items, where $S_m$ denotes the symmetric group of all permutations of $[m]$ and $\Unif(\mathcal{X})$ denotes the uniform distribution over a finite set $\mathcal{X}$. Agent $i$'s randomized valuation $v_i^{\pi_i} : 2^{[m]} \to \R$ is defined on each bundle $S \subseteq [m]$ by
\[
  v_i^{\pi_i}(S) := v_i\!\left(\pi_i^{-1}(S)\right), \qquad \text{where} \qquad \pi_i^{-1}(S) := \{\pi_i^{-1}(g) : g \in S\}.
\]
In other words, the items are ``renamed'' for agent $i$ based on the random permutation $\pi_i$. Note that any high-probability guarantee holding for all base valuations in this semi-random model automatically carries over to the classical i.i.d.\ additive model, since an i.i.d.\ additive valuation is invariant under item renaming.

For an allocation $A = (A_1, \dots, A_n)$, we similarly write $\pi_i^{-1}(A) := (\pi_i^{-1}(A_1), \dots, \pi_i^{-1}(A_n))$. Under $\boldsymbol{\pi} = (\pi_1, \dots, \pi_n) \in S_m^n$, an allocation $A = (A_1, \dots, A_n)$ is envy-free if and only if
\[
  v_i\!\left(\pi_i^{-1}(A_i)\right) \ge v_i\!\left(\pi_i^{-1}(A_j)\right) \qquad \text{for all } i, j \in [n].
\]

Notice that divisibility ($n \mid m$) is necessary for envy-freeness in the random renaming model: if every agent only cares about the number of items she receives, i.e., $v_i(S) = |S|$ for all $i \in [n]$, then random renaming leaves $v_i$ unchanged ($v_i^{\pi_i}(S) = |S|$), and in any allocation when $n \nmid m$, some agent receives strictly fewer items than another and envies her with probability $1$. Throughout this paper, we therefore assume that $n$ divides $m$, write $q := m/n \in \mathbb{Z}_{\ge 1}$ for the equal bundle size, and focus on \emph{balanced allocations} in which every agent receives a bundle of size $q$. We write $\binom{[m]}{q}$ for the collection of all $q$-element subsets of $[m]$, and denote the set of all balanced allocations by
\[
  \cU := \left\{ A = (A_1, \dots, A_n) \in \binom{[m]}{q}^n : A_1 \sqcup \dots \sqcup A_n = [m] \right\}.
\]
We also assume without loss of generality that each base valuation $v_i$ assigns distinct values to all bundles in $\binom{[m]}{q}$.\footnote{If $v_i$ has ties on $\binom{[m]}{q}$, we may perturb $v_i$ infinitesimally to obtain a valuation $\tilde{v}_i$ taking distinct values on $\binom{[m]}{q}$ such that $\tilde{v}_i(S) > \tilde{v}_i(S') \implies v_i(S) \ge v_i(S')$ for all $S, S' \in \binom{[m]}{q}$. Any balanced allocation that is envy-free with respect to $\tilde{v}_1, \dots, \tilde{v}_n$ is then also envy-free with respect to $v_1, \dots, v_n$.}

\paragraph{Our Main Result.}
In~\cite{BHPV24}, Benad\`e et al.\ proved that envy-free allocations exist with high probability in the random renaming model in two special cases: (i) for any number of agents $n$ when valuations are restricted to be \emph{order-consistent} (a class where bundle preferences respect an underlying item ranking) and $m = \omega(n^2)$ with $n \mid m$; and (ii) for arbitrary valuation functions in the special case of $n = 2$ agents and even $m$. However, their $n = 2$ proof relies crucially on the fact that a balanced two-agent allocation $(S, [m] \setminus S)$ is completely determined by a single subset $S \in \binom{[m]}{m/2}$, which fails for $n \ge 3$. They highlighted as their main open question whether EF allocations exist with high probability for general valuation functions when $n > 2$ (noting that the question remained open even for monotone submodular valuations).

We resolve this open question in the affirmative for all $n \ge 3$ and all valuation functions.

\begin{theorem}\label{thm:main}
For any $n \geq 3$, any $m$ divisible by $n$, and any base valuation functions $v_1, \dots, v_n : 2^{[m]} \to \R$, the probability that no envy-free allocation exists in the random renaming model is $O\!\left(\frac{n^2}{m^{1/180}}\right)$.
\end{theorem}

In particular, Theorem~\ref{thm:main} implies that an envy-free allocation exists with high probability whenever $m = \omega(n^{360})$. We have not attempted to optimize the exponent $1/180$ in Theorem~\ref{thm:main}, prioritizing a simpler and self-contained proof instead. Note that a polynomial dependence on $n$ is necessary: as shown by Benad\`e et al.~\cite{BHPV24} via a birthday-paradox construction (discussed in Section~\ref{subsec:related} below), an envy-free allocation can fail to exist with constant probability when $m = O(n^2)$. Closing the polynomial gap between our upper bound on $m$ and this $m = \Omega(n^2)$ lower bound is an interesting open question; see Section~\ref{sec:conclusion} for further discussion.

\subsection{Related Work}\label{subsec:related}

\paragraph{The Random Renaming Model.}
In the random renaming model introduced by Benad\`e et al.~\cite{BHPV24}, their two positive results discussed above rely on techniques specialized to their respective settings that do not extend to $n \ge 3$ arbitrary valuations: for order-consistent valuations, they establish high-probability EF via variants of the Round-Robin algorithm, which requires an underlying single-item ranking that does not exist for general set functions; for $n = 2$ arbitrary valuations, they use a tournament hypergraph symmetry argument over complementary half-bundles $(S, [m] \setminus S)$, which relies on a two-agent allocation being determined by a single subset $S \in \binom{[m]}{m/2}$ and breaks down for $n \ge 3$. In the opposite direction to Theorem~\ref{thm:main}, they also proved a lower bound of $m = \Omega(n^2)$. Specifically, suppose each agent has the unit-demand valuation $v_i(S) = \I[1 \in S]$ (where $\I[\mathcal{E}] \in \{0, 1\}$ denotes the indicator of an event $\mathcal{E}$), deriving value $1$ if her bundle contains her single favorite item $\pi_i(1)$ and $0$ otherwise. Because $\pi_1(1), \dots, \pi_n(1)$ are independent and uniformly distributed in $[m]$, the birthday paradox implies that with probability $\Omega(\min(n^2/m, 1))$, at least two agents share the same favorite item, in which case one of them must receive value $0$. Thus, an EF allocation fails to exist with constant probability whenever $m = O(n^2)$.

\paragraph{Asymptotic Envy-Freeness for Additive Valuations.}
In the classical i.i.d.\ additive model where each agent $i$'s value $v_{i, g}$ for item $g$ is drawn independently from a distribution $\mathcal{D}$, Dickerson et al.~\cite{DGK14} showed that giving each item to the agent who values it most yields an envy-free and Pareto-optimal allocation with high probability when $m = \Omega(n\log n)$, whereas no envy-free allocation exists with high probability when $m = n + o(n)$. Manurangsi and Suksompong~\cite{MS20,MS21} subsequently determined the sharp threshold for envy-freeness: under mild density assumptions on $\mathcal{D}$ (such as bounded density on $[0, 1]$), as $n \to \infty$, when $n \mid m$, EF allocations exist with high probability already for $m \ge 2n$; for general $m$ not ``nearly'' divisible by $n$, EF allocations fail to exist with high probability for $m = O(n\log n / \log\log n)$, and exist with high probability via Round-Robin for $m = \Omega(n\log n / \log\log n)$. Bai and G\"olz~\cite{BG22} extended the $m = \Omega(n\log n / \log\log n)$ envy-freeness and $m = \Omega(n\log n)$ EF + Pareto-optimality upper bounds to agents with independent, non-identical value distributions. Further works on stochastic additive valuations have studied smoothed additive utilities~\cite{BFGP22}, indivisible chores~\cite{MS25a}, unequal agent entitlements (weights)~\cite{MSY25}, online/dynamic arrivals~\cite{ZP20,BHP22,BKP22}, and house allocation or matching settings~\cite{GSV19,YI25,KSSY26}.

\paragraph{Other Fairness Notions and Group Fair Division.}
Beyond envy-freeness among individual agents, asymptotic fair division has also been studied for alternative fairness criteria such as proportionality~\cite{Suk16,MS21,BFGP22,MS25a,MSY25}, maximin share~\cite{KPW16,AMNS17,FGH19,BFGP22}, and EF up to any item (EFX)~\cite{MS21,BHPV24}. Another related line of research investigates asymptotic fair division when items are allocated to predetermined or variable groups of agents~\cite{MS17,GIMS25}; in particular, our proof uses a combinatorial notion of $\delta$-intersection balanced allocations introduced by G\"olz et al.~\cite{GIMS25} in the latter setting.

\subsection{Proof Overview}\label{subsec:overview}

Before turning to the formal proof, we give a high-level overview of the argument. Set the target failure probability $\varepsilon := \frac{30 n^2}{m^{1/180}}$ (assuming $\varepsilon \le 1/2$, as otherwise Theorem~\ref{thm:main} holds trivially); we will show that an envy-free allocation exists with probability at least $1 - \varepsilon$.

\paragraph{The Second Moment Method.}
As noted above, for $n \ge 3$ agents with arbitrary valuation functions $v_i$, neither item-by-item procedures nor two-agent complement symmetry applies, as agent $i$'s envy-freeness condition $v_i(\pi_i^{-1}(A_i)) \ge \max_{j \ne i} v_i(\pi_i^{-1}(A_j))$ depends on all $n$ bundles of a balanced allocation $A = (A_1, \dots, A_n) \in \cU$ simultaneously. Instead, we prove existence via the \emph{second moment method}: we define a non-negative random variable $Z$ counting ``good'' allocations (such that $Z > 0$ implies the existence of an envy-free allocation) and apply the Paley--Zygmund inequality,
\[
  \Prb_{\boldsymbol{\pi}}[Z > 0] \ge \frac{(\E[Z])^2}{\E[Z^2]}.
\]
To make this ratio at least $1 - \varepsilon$, we need $\E[Z^2] \le (1 + \varepsilon)(\E[Z])^2$. It is instructive to first understand how our proof works when $Z$ sums over \emph{all} balanced allocations $A \in \cU$, and then see at the end why we restrict the sum to a carefully selected subcollection of $\cU$.

\paragraph{Top-$\alpha$ Subset Ranking and Good Allocations.}
At first glance, it seems most intuitive to let $Z$ count the number of EF allocations $A \in \cU$. However, the difficulty is that for an arbitrary valuation $v_i$, the probability that a random allocation $B \sim \Unif(\cU)$ is EF for agent $i$ depends on the full structure of $v_i$ and is hard to control.

To overcome this, we introduce our main insight: replacing EF with a \emph{top-$\alpha$ subset-rank condition}. For each agent $i \in [n]$, let
\[
  \cT_i \subseteq \binom{[m]}{q}, \qquad \text{with} \qquad \frac{|\cT_i|}{\binom{m}{q}} = \alpha,
\]
denote the set of the top $\alpha$-fraction of $q$-subsets with the highest values under $v_i$, where $\alpha \approx \left(\frac{\varepsilon}{30n^2}\right)^{18} \in (0, 1)$ is a small target density. We call a balanced allocation $B = (B_1, \dots, B_n) \in \cU$ \emph{good for agent $i$} (i.e., $\cG_i(B) = 1$) if agent $i$'s own bundle $B_i$ falls in her top tier $\cT_i$, while all other $n - 1$ bundles $B_j$ ($j \ne i$) fall outside $\cT_i$:
\[
  \cG_i(B) := \I\!\left[B_i \in \cT_i \;\land\; \forall j \in [n] \setminus \{i\}, \; B_j \notin \cT_i\right].
\]
If $\cG_i(\pi_i^{-1}(A)) = 1$ for every agent $i \in [n]$, then every agent $i$ receives a bundle $\pi_i^{-1}(A_i) \in \cT_i$ that she strictly prefers to every other bundle $\pi_i^{-1}(A_j) \notin \cT_i$ ($j \ne i$), so $A$ is envy-free under $\boldsymbol{\pi}$.

For any fixed allocation $A \in \cU$, because $\pi_1, \dots, \pi_n \sim \Unif(S_m)$ are independent and each $\pi_i^{-1}(A)$ is distributed uniformly over $\cU$, the probability that $A$ is good for all $n$ agents simultaneously is
\[
  \mu := \Prb_{\boldsymbol{\pi}}\!\left[\prod_{i=1}^n \cG_i\!\left(\pi_i^{-1}(A)\right) = 1\right] = \prod_{i=1}^n \E_{B \sim \Unif(\cU)}[\cG_i(B)].
\]
If $\E_{B \sim \Unif(\cU)}[\cG_i(B)] \approx \Prb[B_i \in \cT_i] = \alpha$ for each $i$, then $\mu \approx \alpha^n$.

\paragraph{Reducing Moments to Edge Collisions in Generalized Johnson Graphs.}
To compare the first and second moments, we need to answer two questions: (1) for a single allocation $A \in \cU$, how close is $\E_{B \sim \Unif(\cU)}[\cG_i(B)]$ to $\Prb[B_i \in \cT_i] = \alpha$? and (2) for a pair of allocations $A, A' \in \cU$, how close is the joint success probability $\Prb_{\boldsymbol{\pi}}[\forall i, \, \cG_i(\pi_i^{-1}(A)) = \cG_i(\pi_i^{-1}(A')) = 1]$ to $\mu^2 \approx \alpha^{2n}$?

Both questions reduce to a single graph-theoretic problem on the vertex set $\binom{[m]}{q}$. For each integer $0 \le t < q$ (with $q \le m/2$), the \emph{generalized Johnson graph} $J(m, q, t)$ (Definition~\ref{def:johnson-graph}) is the undirected graph on $\binom{[m]}{q}$ in which two $q$-subsets $S, S'$ are adjacent if and only if $|S \cap S'| = t$; in the special case $t = 0$ (disjoint subsets), $J(m, q, 0)$ is the \emph{Kneser graph}. Because a uniformly random permutation $\pi_i \sim \Unif(S_m)$ preserves intersection sizes and acts transitively on pairs of $q$-subsets of any given overlap $t$, for any fixed $S_0, S'_0 \in \binom{[m]}{q}$ with $|S_0 \cap S'_0| = t$, the randomly renamed pair $(\pi_i^{-1}(S_0), \pi_i^{-1}(S'_0))$ is distributed as a uniformly random oriented edge of the graph $J(m, q, t)$, which is denoted by $(S, S') \sim E(J(m, q, t))$. This translates both moments into \emph{edge collision probabilities} $\Prb_{(S, S') \sim E(J(m, q, t))}[S \in \cT_i \land S' \in \cT_i]$:
\begin{itemize}
  \item \textbf{First moment via the Kneser graph ($t = 0$):} In a single random allocation $B = \pi_i^{-1}(A) \sim \Unif(\cU)$, agent $i$'s bundle $B_i$ lands in $\cT_i$ with probability $\alpha$, while for each $j \ne i$, the disjoint pair $(B_i, B_j)$ is a uniformly random edge of the Kneser graph $J(m, q, 0)$. By a union bound over the $n - 1$ other agents $j \ne i$,
  \[
    \E_{B \sim \Unif(\cU)}[\cG_i(B)] \ge \alpha - (n - 1) \Prb_{(S, S') \sim E(J(m, q, 0))}\!\left[S \in \cT_i \land S' \in \cT_i\right].
  \]
  Thus, to ensure $\E[\cG_i(B)] = (1 - o(1/n))\alpha$ (and hence $\mu = (1 - o(1))\alpha^n$), it suffices to show that the edge collision probability in $J(m, q, 0)$ is $o(\alpha / n^2)$.

  \item \textbf{Second moment via the generalized Johnson graph ($t \approx q/n$):} For two allocations $A, A' \in \cU$ with $t_i := |A_i \cap A'_i| < q$ for each $i \in [n]$, we use the pointwise upper bound $\cG_i(B) \le \I[B_i \in \cT_i]$ (dropping the conditions $B_j \notin \cT_i$ for $j \ne i$). Under this relaxation, if both $A$ and $A'$ are good for all agents, then both $\pi_i^{-1}(A_i) \in \cT_i$ and $\pi_i^{-1}(A'_i) \in \cT_i$ must hold for each $i \in [n]$. Since $(\pi_i^{-1}(A_i), \pi_i^{-1}(A'_i))$ is a uniformly random edge of $J(m, q, t_i)$, by independence of $\pi_1, \dots, \pi_n$, the joint success probability is at most
  \[
    \prod_{i=1}^n \Prb_{(S, S') \sim E(J(m, q, t_i))}\!\left[S \in \cT_i \land S' \in \cT_i\right].
  \]
  Because we relaxed $\cG_i(B)$ to $\I[B_i \in \cT_i]$ (which has expectation $\alpha$), this product should be compared against $\mu^2 \approx \alpha^{2n}$, which is why showing $\E[\cG_i(B)] \approx \alpha$ in the first moment above is crucial. Moreover, for two typical allocations $A, A' \in \cU$, two bundles $A_i, A'_i$ of size $q = m/n$ have expected overlap $q^2/m = q/n$. Thus, to bound the product above by $\alpha^{2n}(1 + O(\varepsilon))$, it suffices to show that whenever $t \approx q/n$, the edge collision probability in $J(m, q, t)$ is at most $\alpha^2(1 + O(\varepsilon / n))$.
\end{itemize}

\paragraph{Bounding Edge Collisions in Kneser and Generalized Johnson Graphs.}
With the reductions above in place, the proof of Theorem~\ref{thm:main} boils down to establishing two purely graph-theoretic edge collision (or \emph{small-set expansion}) bounds for arbitrary vertex subsets $\cT \subseteq \binom{[m]}{q}$ of density $\alpha$:
\begin{enumerate}
  \item In the Kneser graph $J(m, q, 0)$, the edge collision probability is $O(\alpha^{19/18}) \ll \frac{\alpha}{n^2}$ (Lemma~\ref{lem:kneser-collision}).
  \item In the generalized Johnson graph $J(m, q, t)$ with $t \approx q/n$, the edge collision probability is $(1 + o(1/n))\alpha^2$ (Lemma~\ref{lem:johnson-ib-collision}).
\end{enumerate}
We prove both bounds in Section~\ref{sec:fourier-proofs}: for the Kneser graph, we couple a random disjoint edge to a product channel on $\{0, 1\}^m$ and apply standard hypercontractivity on product spaces~\cite{Bonami70,Beckner75,ODonnell14}; for $J(m, q, t)$ with $t \approx q/n$, we combine the Expander Mixing Lemma~\cite{AC88,HLW06} with an elementary bound on the eigenvalues of $J(m, q, t)$~\cite{BCIM18,KR25}. We defer the full proofs of both collision bounds to Section~\ref{sec:fourier-proofs}.

It is worth noting that while Kneser and generalized Kneser graphs have appeared in the fair division literature before, prior works relied on their \emph{topological coloring}~\cite{Suk16agreeable,MS19agreeable,GHIMS22,KSV20,SHS21,FHSZ21} or \emph{local-search}~\cite{PR20} properties rather than their expansion. To our knowledge, our work is the first to leverage the \emph{spectral and (small-set) expansion} properties of Kneser and generalized Johnson graphs in fair division. We hope that this connection can lead to further insight in future research.

\paragraph{Technical Detail: Pruning Atypical Overlaps via $\delta$-IB Allocations.}
There is one remaining technical detail we have glossed over above. If we defined $Z$ as the sum of $\prod_{i=1}^n \cG_i(\pi_i^{-1}(A))$ over \emph{all} balanced allocations $A \in \cU$, then while almost all pairs $(A, A') \in \cU^2$ have bundle overlaps $t_i = |A_i \cap A'_i|$ very close to $q/n = m/n^2$, a tiny fraction of pairs $(A, A')$ have atypical overlaps far from $q/n$ (for example, $A' = A$ where $A_i = A'_i$, or $A'$ differing from $A$ on only a few items where $t_i \approx q$). When $n$ is large and $\mu \approx \alpha^n$ is small, even a small fraction of such highly correlated pairs could inflate $\E[Z^2]$.

To avoid dealing with these atypical pairs, we borrow a technique of G\"olz et al.~\cite{GIMS25}: rather than summing over all of $\cU$, we restrict the count
\[
  Z := \sum_{k=1}^D \prod_{i=1}^n \cG_i\!\left(\pi_i^{-1}(A^{(k)})\right)
\]
to a subcollection $A^{(1)}, \dots,\allowbreak A^{(D)} \in \cU$ of balanced allocations that are \emph{$\delta$-intersection balanced ($\delta$-IB)}, meaning that every pair $A^{(k)}, A^{(\ell)}$ ($k \ne \ell$) satisfies $\big| |A_i^{(k)} \cap A_j^{(\ell)}| - \frac{q}{n} \big| \le \delta \frac{q}{n}$ for all $i, j \in [n]$ (Definition~\ref{def:ib}). Because two independent random allocations satisfy this overlap condition with high probability by a standard concentration inequality, a union bound shows that a $\delta$-IB family of large size exists (Lemma~\ref{lem:ib-existence}). This concludes our proof overview.

\section{Preliminaries}\label{sec:prelim}

Throughout the remainder of the paper, $n \ge 3$ denotes the number of agents, $m$ is a positive integer divisible by $n$, and $q := m/n$ is the bundle size in a balanced allocation.

\subsection{Generalized Johnson and Kneser Graphs}\label{subsec:prelim-graphs}

\begin{definition}[Generalized Johnson Graph and Kneser Graph~(e.g.,~\cite{KMMS25})]\label{def:johnson-graph}
Let $t, q, m$ be integers such that $0 \le t < q \le m/2$. The \emph{generalized Johnson graph} $J(m, q, t)$ is the undirected graph with vertex set $\binom{[m]}{q}$ in which two vertices $S, S' \in \binom{[m]}{q}$ are connected by an edge if and only if $|S \cap S'| = t$. When $t = 0$, the graph $J(m, q, 0)$ is the \emph{Kneser graph}.
\end{definition}

Let $(S, S') \sim E(J(m, q, t))$ denote a uniformly random oriented edge of $J(m, q, t)$, i.e., a pair of subsets $(S, S') \in \binom{[m]}{q} \times \binom{[m]}{q}$ drawn uniformly at random from all pairs satisfying $|S \cap S'| = t$.

\subsection{\texorpdfstring{$\delta$}{delta}-Intersection Balanced Allocations}\label{subsec:prelim-ib}

If two balanced allocations $A, A' \in \cU$ are chosen independently and uniformly at random, then for any agents $i, j \in [n]$, the expected size of the overlap $|A_i \cap A'_j|$ is $q \cdot \frac{q}{m} = \frac{q}{n} = \frac{m}{n^2}$. Following G\"olz et al.~\cite{GIMS25}, we call a collection of balanced allocations \emph{$\delta$-intersection balanced ($\delta$-IB)} if all pairwise bundle intersections lie within a $(1 \pm \delta)$ factor of this expectation:

\begin{definition}[$\delta$-Intersection Balanced Allocations~\cite{GIMS25}]\label{def:ib}
For $\delta \in [0, 1]$, two balanced allocations $A = (A_1, \dots, A_n) \in \cU$ and $A' = (A'_1, \dots, A'_n) \in \cU$ are said to be \emph{$\delta$-intersection balanced ($\delta$-IB)} if
\[
  \frac{m(1 - \delta)}{n^2} \le |A_i \cap A'_j| \le \frac{m(1 + \delta)}{n^2} \qquad \text{for all } i, j \in [n].
\]
We say that balanced allocations $A^{(1)}, \dots, A^{(D)} \in \cU$ are $\delta$-IB if every pair $A^{(k)}, A^{(\ell)}$ ($1 \le k < \ell \le D$) among them is $\delta$-IB.
\end{definition}

The following lemma guarantees the existence of large $\delta$-IB families of balanced allocations whenever $m$ is logarithmic in the family size $D$. Its proof follows that of G\"olz et al.~\cite[Lemma~5]{GIMS25}; we repeat it here because the parameters are slightly different.

\begin{lemma}[cf.~{\cite[Lemma~5]{GIMS25}}]\label{lem:ib-existence}
For any integers $n, D \ge 2$, and any $\delta \in (0, 1)$, if $m$ is divisible by $n$ and $m > \frac{n^3}{\delta^2} \ln(n D)$,
then there exist $D$ balanced allocations $A^{(1)}, \dots, A^{(D)} \in \cU$ that are $\delta$-IB.
\end{lemma}

\begin{proof}
Draw $A^{(1)}, \dots, A^{(D)} \sim \Unif(\cU)$ independently. Fix any $1 \le k < \ell \le D$ and $i, j \in [n]$. Conditioned on $A_i^{(k)}$, the bundle $A_j^{(\ell)} \sim \Unif\binom{[m]}{q}$ ($q = m/n$) is a uniform sample of size $q$ drawn without replacement from $[m]$, so $|A_i^{(k)} \cap A_j^{(\ell)}|$ is the sum of $q$ draws without replacement from the $\{0, 1\}$-population $(\I[g \in A_i^{(k)}])_{g \in [m]}$ of mean $\frac{q}{m} = \frac{1}{n}$. By Hoeffding's inequality for sampling without replacement~\cite[Theorems~2 and~4]{Hoeffding63},
\[
  \Prb\!\left[\left| |A_i^{(k)} \cap A_j^{(\ell)}| - \frac{m}{n^2} \right| > \frac{\delta m}{n^2}\right] \le 2 \exp\!\left(-\frac{2(\delta m / n^2)^2}{m / n}\right) = 2 \exp\!\left(-\frac{2\delta^2 m}{n^3}\right).
\]
By a union bound over all $\binom{D}{2}$ pairs $\{k, \ell\}$ and all $n^2$ pairs $(i, j) \in [n]^2$, the probability that $A^{(1)}, \dots, A^{(D)}$ fails to be $\delta$-IB is less than $n^2 D^2 \exp(-2\delta^2 m / n^3) \le 1$, so a $\delta$-IB collection exists.
\end{proof}

\section{Proof of Main Theorem via Graph Collision Bounds}\label{sec:main-proof}

In this section, we state our two main graph-theoretic lemmas, which bound the probability that both endpoints of a random edge in the Kneser graph $J(m, q, 0)$ or the generalized Johnson graph $J(m, q, t)$ fall into a given vertex subset $\cT \subseteq \binom{[m]}{q}$, and show how they immediately imply Theorem~\ref{thm:main}. We then prove both graph collision lemmas in Section~\ref{sec:fourier-proofs}.

\subsection{Edge Collision Bounds in the Kneser and Generalized Johnson Graphs}\label{subsec:collision-lemmas}

Suppose $\cT \subseteq \binom{[m]}{q}$ is any set of vertices of density $\alpha := \frac{|\cT|}{\binom{m}{q}} \in (0, 1)$. If we sample two vertices $S, S' \sim \Unif\binom{[m]}{q}$ \emph{independently}, each lands in $\cT$ with probability $\alpha$, so both land in $\cT$ with probability $\alpha^2$. What if instead $(S, S') \sim E(J(m, q, t))$ is a uniformly random \emph{edge} of $J(m, q, t)$? In our lemmas below, we establish an upper bound on this ``collision probability'' for two regimes of the parameter $t$. Another way to interpret these bounds is that the generalized Johnson graphs (for these values of $t$) are \emph{(small-set) expanders}, i.e., no subset $\cT$ of density $\alpha$ contains too many edges of $J(m, q, t)$.

Our first lemma handles the Kneser graph $J(m, q, 0)$, and gives an upper bound of $O(\alpha^{19/18})$ on the collision probability.

\begin{lemma}\label{lem:kneser-collision}
For $m \ge 3q$ and any $\cT \subseteq \binom{[m]}{q}$ of density $\alpha := \frac{|\cT|}{\binom{m}{q}} \in (0, 1)$,
\[
  \Prb_{(S, S') \sim E(J(m, q, 0))}\!\left[S \in \cT \;\land\; S' \in \cT\right] \le 5 \, \alpha^{19/18}.
\]
\end{lemma}

Our second lemma handles the generalized Johnson graph $J(m, q, t)$ when the overlap $t$ is close to the expected overlap $q/n = m/n^2$ of two independent $q$-subsets, as guaranteed by a $\delta$-IB allocation family. It shows that in this regime, the edge collision probability in $J(m, q, t)$ is very close to the independent value $\alpha^2$:

\begin{lemma}\label{lem:johnson-ib-collision}
Let $t \in \{0, 1, \dots, q - 1\}$ and $\delta \in (0, 1/6]$ be such that $|t - q/n| \le \delta \frac{q}{n}$ (where $m = n q$ with $n \ge 3$). For any $\cT \subseteq \binom{[m]}{q}$ of density $\alpha := \frac{|\cT|}{\binom{m}{q}} \in (0, 1)$,
\[
  \Prb_{(S, S') \sim E(J(m, q, t))}\!\left[S \in \cT \;\land\; S' \in \cT\right] \le \alpha^2 + \alpha\!\left(\delta + \frac{50 n}{m}\right).
\]
\end{lemma}

\subsection{Proof of Theorem~\ref{thm:main} from Lemmas~\ref{lem:kneser-collision} and~\ref{lem:johnson-ib-collision}}\label{subsec:main-proof-derivation}

We now show how Lemmas~\ref{lem:ib-existence}, \ref{lem:kneser-collision}, and~\ref{lem:johnson-ib-collision} combine via the second moment method to prove Theorem~\ref{thm:main}.

\paragraph{Parameters and Setup.}
Given $n \ge 3$ and $m = n q$, set
\begin{equation}\label{eq:eps-choice}
  \varepsilon := \frac{30 n^2}{m^{1/180}}.
\end{equation}
Note that if $\varepsilon > 1/2$, Theorem~\ref{thm:main} holds trivially (as $1 < 2\varepsilon$). We may therefore assume henceforth that $\varepsilon \in (0, 1/2]$ (which in particular implies $m \ge (60 n^2)^{180} \ge 540^{180}$), and show that the failure probability is at most $\varepsilon$.

Let $K := \left\lfloor m^{-1/10} \binom{m}{q} \right\rfloor$ and define the top-tier density $\alpha := K / \binom{m}{q}$. Since $1 \le q = m/n \le m - 1$, we have $\binom{m}{q} \ge m$ and hence $m^{-1/10}\binom{m}{q} \ge m^{9/10} > 2$; therefore $1 \le K < \binom{m}{q}$, and $\alpha \in (0, 1)$ satisfies
\begin{equation}\label{eq:alpha-sandwich}
  \frac{1}{2} m^{-1/10} \le \alpha \le m^{-1/10} = \left(\frac{\varepsilon}{30 n^2}\right)^{18}.
\end{equation}

Recall that, without loss of generality, each $v_i$ takes distinct values on $\binom{[m]}{q}$. For each agent $i \in [n]$, let $\cT_i \subseteq \binom{[m]}{q}$ with $|\cT_i| = K = \alpha \binom{m}{q}$ denote the set of the $K$ subsets from $\binom{[m]}{q}$ with the largest values under $v_i$. For any balanced allocation $B = (B_1, \dots, B_n) \in \cU$, we define the indicator $\cG_i(B) \in \{0, 1\}$ of the event that $B$ is \emph{good for agent $i$} by
\begin{equation}\label{eq:Gi-def}
  \cG_i(B) := \I\!\left[B_i \in \cT_i \;\land\; \forall j \in [n] \setminus \{i\}, \; B_j \notin \cT_i\right].
\end{equation}
Notice that this event suffices for envy-freeness: for any allocation $A = (A_1, \dots, A_n) \in \cU$ and any permutations $\boldsymbol{\pi} = (\pi_1, \dots, \pi_n) \in S_m^n$, if $\cG_i(\pi_i^{-1}(A)) = 1$, then agent $i$'s renamed bundle $\pi_i^{-1}(A_i)$ belongs to her top tier $\cT_i$, whereas every other agent's renamed bundle $\pi_i^{-1}(A_j)$ ($j \ne i$) lies outside $\cT_i$. Since every subset in $\cT_i$ has strictly larger value under $v_i$ than every subset outside $\cT_i$, we have $v_i(\pi_i^{-1}(A_i)) > v_i(\pi_i^{-1}(A_j))$ for all $j \ne i$. Therefore:
\begin{equation}\label{eq:G-implies-EF}
  \prod_{i=1}^n \cG_i\!\left(\pi_i^{-1}(A)\right) = 1 \quad \Longrightarrow \quad A \text{ is envy-free under } \boldsymbol{\pi}.
\end{equation}

\paragraph{First-Moment Lower Bound via Lemma~\ref{lem:kneser-collision}.}
We first use the Kneser graph collision bound (Lemma~\ref{lem:kneser-collision}) to show that for any fixed balanced allocation $A \in \cU$, the probability that $\cG_i(\pi_i^{-1}(A)) = 1$ for all $i \in [n]$ is very close to $\alpha^n$:

\begin{lemma}[First-Moment Bound]\label{lem:first-moment}
Let $\varepsilon \in (0, 1/2]$ and $\alpha \in (0, 1)$ be as in~\eqref{eq:eps-choice} and~\eqref{eq:alpha-sandwich}. For any fixed balanced allocation $A \in \cU$, when $\pi_1, \dots, \pi_n \sim \Unif(S_m)$ are drawn independently, the probability
\[
  \mu := \Prb_{\boldsymbol{\pi} \sim \Unif(S_m^n)}\!\left[\prod_{i=1}^n \cG_i\!\left(\pi_i^{-1}(A)\right) = 1\right] = \prod_{i=1}^n \E_{B \sim \Unif(\cU)}[\cG_i(B)]
\]
does not depend on $A$ and satisfies
\begin{equation}\label{eq:mu-lower}
  \mu \ge \left(1 - \frac{\varepsilon}{6n}\right)^n \alpha^n \ge \frac{1}{2}\alpha^n, \qquad \left(\frac{\alpha^n}{\mu}\right)^2 \le \frac{1}{1 - \varepsilon/3} \le 1 + \frac{5\varepsilon}{12}.
\end{equation}
\end{lemma}

\begin{proof}
Since $\pi_1, \dots, \pi_n$ are independent and each $\pi_i^{-1}(A)$ is distributed uniformly over $\cU$, the events $\{\cG_i(\pi_i^{-1}(A)) = 1\}_{i=1}^n$ are independent, with respective probabilities $\E_{B \sim \Unif(\cU)}[\cG_i(B)]$. Fix any agent $i \in [n]$ and draw $B = (B_1, \dots, B_n) \sim \Unif(\cU)$. Then $B_i \sim \Unif\binom{[m]}{q}$ has $\Prb[B_i \in \cT_i] = \frac{|\cT_i|}{\binom{m}{q}} = \alpha$. Moreover, for each $j \in [n] \setminus \{i\}$, the bundles $B_i$ and $B_j$ are disjoint $q$-subsets, so the pair $(B_i, B_j)$ is a uniformly random edge of the Kneser graph $J(m, q, 0)$. By the union bound over $j \in [n] \setminus \{i\}$ and Lemma~\ref{lem:kneser-collision} applied to $\cT = \cT_i$,
\begin{align*}
  \E_{B \sim \Unif(\cU)}[\cG_i(B)]
  &= \Prb_{B \sim \Unif(\cU)}\!\left[B_i \in \cT_i \;\land\; \forall j \ne i, \; B_j \notin \cT_i\right] \\
  &\ge \Prb\!\left[B_i \in \cT_i\right] - \sum_{j \ne i} \Prb_{(B_i, B_j) \sim E(J(m, q, 0))}\!\left[B_i \in \cT_i \;\land\; B_j \in \cT_i\right] \\
  &\ge \alpha - (n - 1) \cdot 5\alpha^{19/18} = \alpha\left(1 - 5(n - 1)\alpha^{1/18}\right) \\
  &\ge \alpha\left(1 - \frac{\varepsilon}{6n}\right),
\end{align*}
where the last inequality uses $\alpha^{1/18} \le \frac{\varepsilon}{30n^2}$ from~\eqref{eq:alpha-sandwich}. Multiplying over $i \in [n]$ gives $\mu \ge (1 - \frac{\varepsilon}{6n})^n \alpha^n \ge (1 - \frac{\varepsilon}{6})\alpha^n \ge \frac{1}{2}\alpha^n$ and $(\mu / \alpha^n)^2 \ge (1 - \frac{\varepsilon}{6n})^{2n} \ge 1 - \frac{\varepsilon}{3}$, which yields~\eqref{eq:mu-lower} for $\varepsilon \in (0, 1/2]$.
\end{proof}

\paragraph{Second-Moment Upper Bound and Completion of the Proof of Theorem~\ref{thm:main}.}
We now combine Lemma~\ref{lem:ib-existence}, Lemma~\ref{lem:johnson-ib-collision}, and Lemma~\ref{lem:first-moment} with the Paley--Zygmund inequality.

\begin{proof}[Proof of Theorem~\ref{thm:main}]
Let $\varepsilon = \frac{30 n^2}{m^{1/180}} \in (0, 1/2]$, $\alpha$, and $\mu$ be as in~\eqref{eq:eps-choice}, \eqref{eq:alpha-sandwich}, and Lemma~\ref{lem:first-moment}. Set
\[
  D := \left\lceil \frac{12}{\varepsilon \mu} \right\rceil \le \left\lceil \frac{24}{\varepsilon \alpha^n} \right\rceil \qquad \text{and} \qquad \delta := \frac{\varepsilon \alpha}{16 n} \le \frac{1}{6},
\]
so that $\frac{1}{D\mu} \le \frac{\varepsilon}{12}$. Let us verify that the condition $m > \frac{n^3}{\delta^2}\ln(n D)$ of Lemma~\ref{lem:ib-existence} holds. By~\eqref{eq:eps-choice} and~\eqref{eq:alpha-sandwich}, we have $\delta = \frac{\varepsilon \alpha}{16 n} \ge \frac{30 n^2}{16 n} m^{-1/180} \cdot \frac{1}{2} m^{-1/10} > \frac{1}{2} n m^{-19/180}$, so $\frac{n^3}{\delta^2} < 4 n m^{19/90}$. Moreover, since $\varepsilon \alpha^n > 25 n \, m^{-n}$, we have $n D \le n\!\left(\frac{24}{\varepsilon \alpha^n} + 1\right) \le m^n$ and thus $\ln(n D) \le n \ln m$. Using $4 n^2 < 60 n^2 \le m^{1/180}$ (from $\varepsilon \le 1/2$), we obtain
\[
  \frac{n^3}{\delta^2} \ln(n D) < 4 n^2 m^{19/90} \ln m < m^{1/180 + 38/180} \ln m < m^{1/4} \ln m < m.
\]

By Lemma~\ref{lem:ib-existence}, there exists a $\delta$-IB collection of balanced allocations $A^{(1)}, \dots, A^{(D)} \in \cU$. For each $k \in [D]$, let $\mathcal{E}_k$ denote the event that $\prod_{i=1}^n \cG_i(\pi_i^{-1}(A^{(k)})) = 1$, and define the counting random variable $Z := \sum_{k=1}^D \I[\mathcal{E}_k]$.

By~\eqref{eq:G-implies-EF}, whenever $Z > 0$, at least one allocation $A^{(k)}$ is envy-free under $\boldsymbol{\pi}$. By Lemma~\ref{lem:first-moment}, $\Prb_{\boldsymbol{\pi}}[\mathcal{E}_k] = \mu$ for every $k \in [D]$, so the first moment of $Z$ is $\E[Z] = D \mu$.

To bound the second moment $\E[Z^2] = \sum_{k=1}^D \Prb[\mathcal{E}_k] + \sum_{k \ne \ell} \Prb[\mathcal{E}_k \land \mathcal{E}_\ell]$, consider any two distinct indices $k \ne \ell \in [D]$. Observe from~\eqref{eq:Gi-def} that $\cG_i(B) \le \I[B_i \in \cT_i]$ holds pointwise (obtained by simply dropping the conditions $B_j \notin \cT_i$ for $j \ne i$). Therefore, if $\mathcal{E}_k \land \mathcal{E}_\ell$ occurs, then for every agent $i \in [n]$, both $\pi_i^{-1}(A_i^{(k)}) \in \cT_i$ and $\pi_i^{-1}(A_i^{(\ell)}) \in \cT_i$ must hold. Since $\pi_1, \dots, \pi_n \sim \Unif(S_m)$ are independent and $\pi_i^{-1}$ maps the pair $(A_i^{(k)}, A_i^{(\ell)})$ to a uniformly random edge $(S, S') \sim E(J(m, q, t_i))$ where $t_i := |A_i^{(k)} \cap A_i^{(\ell)}|$, we obtain
\begin{align}\label{eq:pair-factor}
  \Prb_{\boldsymbol{\pi}}\!\left[\mathcal{E}_k \land \mathcal{E}_\ell\right] &\le \prod_{i=1}^n \Prb_{\pi_i}\!\left[\pi_i^{-1}(A_i^{(k)}) \in \cT_i \;\land\; \pi_i^{-1}(A_i^{(\ell)}) \in \cT_i\right] \nonumber \\
  &= \prod_{i=1}^n \Prb_{(S, S') \sim E(J(m, q, t_i))}\!\left[S \in \cT_i \;\land\; S' \in \cT_i\right].
\end{align}
Because $A^{(k)}$ and $A^{(\ell)}$ are $\delta$-IB, each overlap $t_i = |A_i^{(k)} \cap A_i^{(\ell)}|$ satisfies $|t_i - q/n| \le \delta \frac{q}{n}$. Moreover, by~\eqref{eq:eps-choice}, \eqref{eq:alpha-sandwich}, and $m \ge 540^{180}$, we have
\[
  \frac{50 n}{\alpha m} \le \frac{100 n}{m^{9/10}} < \frac{n}{m^{1/10}} < \frac{\varepsilon}{16 n}.
\]
Applying Lemma~\ref{lem:johnson-ib-collision} to each $\cT_i$ (which has exact density $\alpha$) therefore gives, for every $i \in [n]$,
\begin{align*}
  \Prb_{(S, S') \sim E(J(m, q, t_i))}\!\left[S \in \cT_i \;\land\; S' \in \cT_i\right] 
  \le \alpha^2\!\left(1 + \frac{\delta}{\alpha} + \frac{50 n}{\alpha m}\right)
  \le \alpha^2\!\left(1 + \frac{\varepsilon}{8 n}\right).
\end{align*}
Substituting this bound into~\eqref{eq:pair-factor} yields
\[
  \Prb_{\boldsymbol{\pi}}\!\left[\mathcal{E}_k \land \mathcal{E}_\ell\right] \le \alpha^{2n}\!\left(1 + \frac{\varepsilon}{8n}\right)^n \le \alpha^{2n} e^{\varepsilon / 8} \le \alpha^{2n}\!\left(1 + \frac{\varepsilon}{6}\right).
\]
Summing over all $D(D-1) < D^2$ off-diagonal pairs $(k, \ell)$ gives
\[
  \E[Z^2] = D\mu + \sum_{k \ne \ell} \Prb_{\boldsymbol{\pi}}\!\left[\mathcal{E}_k \land \mathcal{E}_\ell\right] \le D\mu + D^2 \alpha^{2n}\!\left(1 + \frac{\varepsilon}{6}\right).
\]
Finally, by the Paley--Zygmund inequality and Lemma~\ref{lem:first-moment},
\begin{align*}
  \Prb_{\boldsymbol{\pi}}\!\left[\text{an EF allocation exists}\right] \ge \Prb_{\boldsymbol{\pi}}[Z > 0] \ge \frac{(\E[Z])^2}{\E[Z^2]} \ge \left(\frac{1}{D\mu} + \left(\frac{\alpha^n}{\mu}\right)^2 \!\left(1 + \frac{\varepsilon}{6}\right)\right)^{-1}.
\end{align*}
Using $\frac{1}{D\mu} \le \frac{\varepsilon}{12}$ and $(\alpha^n / \mu)^2 \le 1 + \frac{5\varepsilon}{12}$ from~\eqref{eq:mu-lower}, the right-hand side is at least
\[
  \left(\frac{\varepsilon}{12} + \left(1 + \frac{5\varepsilon}{12}\right)\!\left(1 + \frac{\varepsilon}{6}\right)\right)^{-1} \ge \frac{1}{1 + \varepsilon} \ge 1 - \varepsilon.
\]
Consequently, the probability that no envy-free allocation exists is at most $\varepsilon = \frac{30 n^2}{m^{1/180}} = O\!\left(\frac{n^2}{m^{1/180}}\right)$, completing the proof of Theorem~\ref{thm:main}.
\end{proof}

\section{Proofs of the Graph Collision Lemmas}\label{sec:fourier-proofs}

It remains to prove Lemmas~\ref{lem:kneser-collision} and \ref{lem:johnson-ib-collision}. In Section~\ref{subsec:proof-kneser}, we prove Lemma~\ref{lem:kneser-collision} by embedding a random edge of $J(m, q, 0)$ into an i.i.d.\ disjoint-subset channel on $\{0, 1\}^m$ and applying standard hypercontractivity on product spaces~\cite{Bonami70,Beckner75,ODonnell14}; in Section~\ref{subsec:proof-johnson-ib}, we prove Lemma~\ref{lem:johnson-ib-collision} directly from the Expander Mixing Lemma~\cite{AC88,HLW06} and the known spectrum of $J(m, q, t)$~\cite{BCIM18,KR25}.

\subsection{Proof of Lemma~\ref{lem:kneser-collision}}\label{subsec:proof-kneser}

One can obtain small-set expansion in $J(m, q, 0)$ via logarithmic Sobolev and hypercontractive inequalities on the Johnson slice $\binom{[m]}{q}$~\cite{LY98,DSC96,FOW19,KK20}. Here we instead give a simpler reduction directly to standard hypercontractivity on an i.i.d.\ product space~\cite{Bonami70,Beckner75,ODonnell14}.

\begin{proposition}[Hypercontractivity Theorem~{\cite[Corollary~10.20]{ODonnell14}}]\label{prop:hypercontractivity}
Let $(\Omega, \gamma)$ be a finite probability space with $|\Omega| \ge 2$ in which every outcome has probability at least $\eta > 0$, with associated $L^p$ norm $\|f\|_p := (\E_{x \sim \gamma^{\otimes m}}[|f(x)|^p])^{1/p}$ ($p \ge 1$). For $\rho \in [0, 1]$, let $(T_\rho f)(x) := \E_y[f(y)]$ denote the noise operator, where each coordinate $y_g$ independently equals $x_g$ with probability $\rho$ and is drawn from $\gamma$ with probability $1 - \rho$. Then for every $s > 2$, every $0 \le \rho \le \frac{1}{\sqrt{s - 1}}\,\eta^{\frac{1}{2} - \frac{1}{s}}$, and every $f : \Omega^m \to \R$,
\[
  \|T_\rho f\|_2 \le \|f\|_{s'}, \qquad \text{where } s' := \frac{s}{s - 1}.
\]
\end{proposition}

\begin{proof}[Proof of Lemma~\ref{lem:kneser-collision}]
Let $\mathcal{Q}$ be the symmetric distribution on $\{0, 1\}^2$ defined by
\[
  \mathcal{Q}(1, 0) = \frac{4}{9}, \quad \mathcal{Q}(0, 1) = \frac{4}{9}, \quad \mathcal{Q}(0, 0) = \frac{1}{9}, \quad \mathcal{Q}(1, 1) = 0.
\]
Draw $(x_1, y_1), \dots, (x_m, y_m) \stackrel{\text{i.i.d.}}{\sim} \mathcal{Q}$, so each coordinate has marginal $\gamma := \mathrm{Bern}(4/9)$, and the subsets $U := \{g \in [m] : x_g = 1\}$ and $W := \{g \in [m] : y_g = 1\}$ are always disjoint (since $\mathcal{Q}(1, 1) = 0$). Identifying subsets of $[m]$ with their $\{0, 1\}^m$ indicator vectors, define $f : \{0, 1\}^m \to [0, 1]$ by
\[
  f(U) := \begin{cases}
    \Prb_{S \sim \Unif\binom{U}{q}}\!\left[S \in \cT\right] & \text{if } |U| \ge q, \\
    0 & \text{if } |U| < q.
  \end{cases}
\]
By symmetry of $\mathcal{Q}^{\otimes m}$, $S \sim \Unif\binom{U}{q}$ conditioned on $|U| \ge q$ is uniformly distributed over $\binom{[m]}{q}$, so $\E[f(U)] = \Prb[|U| \ge q] \cdot \alpha \le \alpha$. Conditioned on $|U| \ge q$ and $|W| \ge q$, independent draws $S \sim \Unif\binom{U}{q}$ and $S' \sim \Unif\binom{W}{q}$ form a uniformly random edge of $J(m, q, 0)$, giving
\begin{align}\label{eq:channel-identity}
  \E\big[f(U) f(W)\big] &= \Prb\big[|U| \ge q \;\land\; |W| \ge q\big] \cdot \Prb_{(S, S') \sim E(J(m, q, 0))}\!\left[S \in \cT \;\land\; S' \in \cT\right].
\end{align}
Since $|U| \sim \mathrm{Bin}(m, 4/9)$ and $m \ge 3q$, we have
\begin{align*}
\Prb[|U| \leq q - 1] \leq \E\!\left[\left(\frac{2}{5}\right)^{|U| - (q - 1)}\right] = \left(\frac{2}{5}\right)^{1 - q}\!\left(\frac{5}{9} + \frac{4}{9}\cdot\frac{2}{5}\right)^m \le \frac{2}{5}\left(\frac{5}{2}\cdot\frac{11^3}{15^3}\right)^q < \frac{2}{5},
\end{align*}
and symmetrically $\Prb[|W| \le q - 1] < \frac{2}{5}$, so $\Prb[|U| \ge q \land |W| \ge q] > 1 - \frac{4}{5} = \frac{1}{5}$ by a union bound.

To upper bound $\E[f(U) f(W)]$, draw $x' = (x'_1, \dots, x'_m)$ conditionally independent of $x$ given $y$ with $(x', y) \sim \mathcal{Q}^{\otimes m}$, and let $U' := \{g \in [m] : x'_g = 1\}$. Then $\E[f(U) \mid W] = \E[f(U') \mid W]$, so $\E[\E[f(U) \mid W]^2] = \E[f(U) f(U')]$, and Cauchy--Schwarz inequality gives
\[
  \E\big[f(U) f(W)\big] = \E\big[\E[f(U) \mid W] \, f(W)\big] \le \|f\|_2 \sqrt{\E\big[f(U) f(U')\big]}.
\]
In each coordinate $g \in [m]$, both $x_g$ and $x'_g$ have marginal $\gamma = \mathrm{Bern}(4/9)$ (with minimum atom probability $\eta = \frac{4}{9}$), and since $x_g = x'_g = 1$ requires $y_g = 0$ (probability $\frac{5}{9}$), we have
\[
  \Prb[x_g = 1 \;\land\; x'_g = 1] = \frac{5}{9}\left(\frac{4}{5}\right)^2 = \frac{16}{45} = \left(\frac{4}{5}\right)^2 \cdot\frac{4}{9} + \left(1 - \left(\frac{4}{5}\right)^2 \right)\left(\frac{4}{9}\right)^2.
\]
Hence $(x, x')$ may be generated by drawing $z \sim \gamma^{\otimes m}$ and applying $T_{4/5}$ independently twice to obtain $x$ and $x'$; thus $\E[f(U) f(U')] = \|T_{4/5} f\|_2^2$. Applying Proposition~\ref{prop:hypercontractivity} with $s = \frac{9}{4}$ ($s' = \frac{9}{5}$), since
$\frac{1}{\sqrt{9/4 - 1}}\left(\frac{4}{9}\right)^{\frac{1}{2} - \frac{4}{9}} > \frac{4}{5}$,
we obtain $\|T_{4/5} f\|_2 \le \|f\|_{9/5}$. Finally, since $f \in [0, 1]$ satisfies $\E[f] \le \alpha$, we have $\|f\|_2 \le \alpha^{1/2}$ and $\|f\|_{9/5} \le \alpha^{5/9}$. Combining with~\eqref{eq:channel-identity} yields $\frac{1}{5}\Prb[S \in \cT \land S' \in \cT] \le \E[f(U) f(W)] \le \alpha^{1/2}\alpha^{5/9} = \alpha^{19/18}$.
\end{proof}

\subsection{Proof of Lemma~\ref{lem:johnson-ib-collision}}\label{subsec:proof-johnson-ib}

We prove Lemma~\ref{lem:johnson-ib-collision} by combining the Expander Mixing Lemma with an elementary bound on the eigenvalues of $J(m, q, t)$ when $t \approx q/n$ (cf.~\cite{BCIM18,KR25}). For a simple regular undirected graph $H$ on $N$ vertices with positive degree, let $1 = \theta_1 \ge \theta_2 \ge \dots \ge \theta_N$ denote the eigenvalues of its normalized adjacency matrix, and set $\lambda(H) := \max_{2 \le c \le N} |\theta_c|$.

\begin{proposition}[Expander Mixing Lemma{~\cite[Lemma~2.3]{AC88}}]\label{prop:eml}
For any simple regular undirected graph $H$ on $N$ vertices with positive degree and any vertex subset $\cT$ of density $\alpha := |\cT| / N$,
\begin{align*}
  \Prb_{(S, S') \sim E(H)}\!\left[S \in \cT \;\land\; S' \in \cT\right]
  &\le \alpha^2 + \lambda(H)\,\alpha(1 - \alpha).
\end{align*}
\end{proposition}

The spectrum of $J(m, q, t)$ is given explicitly by the Eberlein polynomials:\footnote{Strictly speaking, {\cite[Section~3]{BCIM18}} parameterizes $J(m, q, t)$ by distance $q - t$ and states the \emph{unnormalized} eigenvalues. Dividing the third formula in {\cite[Section~3]{BCIM18}} by the degree $\binom{q}{q - t}\binom{m - q}{q - t}$ of $J(m, q, t)$ yields the expression in Proposition~\ref{prop:eigenvalues}.}

\begin{proposition}[Spectrum of $J(m, q, t)$~(e.g.,~\cite{BCIM18})]\label{prop:eigenvalues}
Let $t, q, m$ be integers such that $0 \le t < q \le m/2$. Counting eigenvalues with multiplicity, the spectrum of the normalized adjacency matrix of $J(m, q, t)$ consists of $\lambda_r(m, q, t)$ with multiplicity $\binom{m}{r} - \binom{m}{r - 1}$ for $0 \le r \le q$ (where $\binom{m}{-1} := 0$), given by $\lambda_0(m, q, t) = 1$ and, for each $1 \le r \le q$,
\begin{align}\label{eq:eigenvalue-formula}
  \lambda_r(m, q, t)
  = \sum_{h=0}^r (-1)^{r - h} \binom{r}{h} \frac{\binom{t}{h}\binom{q - t}{r - h}}{\binom{q}{h}\binom{m - q}{r - h}} = \sum_{h=0}^r (-1)^{r - h} \binom{r}{h} \frac{(t)_h}{(q)_h} \frac{(q - t)_{r - h}}{(m - q)_{r - h}},
\end{align}
where, for any integers $z, d \ge 0$, $(z)_d := \prod_{w=0}^{d - 1} (z - w)$ denotes\footnote{Note that $(z)_d = 0$ whenever $d > z$.} the falling factorial (with $(z)_0 = 1$).
\end{proposition}

We can now prove Lemma~\ref{lem:johnson-ib-collision}.

\begin{proof}[Proof of Lemma~\ref{lem:johnson-ib-collision}]
Since $n \ge 3$ and $\delta \in (0, 1/6]$, we have $q = m/n < m/2$ and $0 < \frac{1 - \delta}{n} \le p_1 := \frac{t}{q} \le \frac{1 + \delta}{n} \le \frac{7}{18} < \frac{1}{2}$. Setting $p_2 := \frac{q - t}{m - q} \in (0, 1)$, we obtain
\begin{align*}
  |p_1 - p_2| &= \frac{n}{n - 1}\frac{|t - q/n|}{q} \le \frac{\delta}{n - 1} < \delta
  \qquad\text{and}\qquad
  p_1 + p_2 = \frac{t}{q} + \frac{q - t}{(n - 1)q} \le \frac{1}{2} + \frac{t}{2 q} < \frac{3}{4}.
\end{align*}

Next, let $a, b$ be integers such that $1 \le a < b$, and set $p := a / b \in (0, 1)$. Note that $\frac{(a)_\ell}{(b)_\ell} = 0$ for every integer $\ell \in \{a + 1, \dots, b\}$, whereas for every integer $0 \le \ell \le a$ we have $0 < \frac{(a)_\ell}{(b)_\ell} \le p^\ell$ and, since $1 \le b - a \le b - \ell$,
\[
  0 \le p - \frac{a - \ell}{b - \ell} = \frac{\ell(b - a)}{b(b - \ell)} \le \frac{\ell}{b}.
\]
Hence, for every integer $1 \le s \le b$, each summand below with $\ell \le a$ lies in $\big[0, \frac{\ell}{b}\, p^{s - 1}\big]$ (while any summand with $\ell > a$ is zero), and telescoping gives
\begin{align}\label{eq:falling-factorial-bound}
  0 \le p^s - \frac{(a)_s}{(b)_s}
  &= \sum_{\ell=0}^{s - 1} p^{s - 1 - \ell}\!\left(p - \frac{a - \ell}{b - \ell}\right)\frac{(a)_\ell}{(b)_\ell}
   = \sum_{\ell=0}^{\min(s - 1, \, a)} p^{s - 1 - \ell}\!\left(p - \frac{a - \ell}{b - \ell}\right)\frac{(a)_\ell}{(b)_\ell} \nonumber \\
  &\le \sum_{\ell=0}^{\min(s - 1, \, a)} \frac{\ell}{b}\, p^{s - 1}
   \le \frac{\binom{s}{2}}{b}\, p^{s - 1}.
\end{align}

Now fix any integer $1 \le r \le q$. Subtracting $(p_1 - p_2)^r = \sum_{h=0}^r (-1)^{r - h}\binom{r}{h} p_1^h p_2^{r - h}$ from~\eqref{eq:eigenvalue-formula} yields
\begin{align}
  \left|\lambda_r(m, q, t) - (p_1 - p_2)^r\right|
  &\le \sum_{h=0}^r \binom{r}{h} \left(p_1^h p_2^{r - h} - \frac{(t)_h}{(q)_h}\frac{(q - t)_{r - h}}{(m - q)_{r - h}}\right) \nonumber \\
  &\le \sum_{h=1}^r \binom{r}{h} \left(p_1^h - \frac{(t)_h}{(q)_h}\right) p_2^{r - h}
     + \sum_{h=0}^{r - 1} \binom{r}{h} p_1^h\!\left(p_2^{r - h} - \frac{(q - t)_{r - h}}{(m - q)_{r - h}}\right) \nonumber \\
  &\le \sum_{h=1}^r \binom{r}{h} \frac{\binom{h}{2}}{q}\, p_1^{h - 1} p_2^{r - h}
     + \sum_{h=0}^{r - 1} \binom{r}{h} \frac{\binom{r - h}{2}}{q}\, p_1^h p_2^{r - h - 1} \nonumber \\
  &\le \frac{\binom{r}{2}}{q} \sum_{h=1}^r \binom{r - 1}{h - 1} p_1^{h - 1} p_2^{r - h}
     + \frac{\binom{r}{2}}{q} \sum_{h=0}^{r - 1} \binom{r - 1}{h} p_1^h p_2^{r - h - 1} \nonumber \\
  &= \frac{r(r - 1)}{q}\,(p_1 + p_2)^{r - 1}, \label{eq:eigen-diff-bound}
\end{align}
where the first two inequalities use $0 \le \zeta_1 \zeta_2 \le \xi_1 \xi_2$ and $\xi_1 \xi_2 - \zeta_1 \zeta_2 \le (\xi_1 - \zeta_1)\xi_2 + \xi_1(\xi_2 - \zeta_2)$ for any reals $0 \le \zeta_1 \le \xi_1$ and $0 \le \zeta_2 \le \xi_2$ (noting that $p_1^0 - \frac{(t)_0}{(q)_0} = p_2^0 - \frac{(q - t)_0}{(m - q)_0} = 0$), the third follows from~\eqref{eq:falling-factorial-bound}, and the fourth uses $\binom{r}{h}\binom{h}{2} = \frac{r(h - 1)}{2}\binom{r - 1}{h - 1} \le \binom{r}{2}\binom{r - 1}{h - 1}$ for $1 \le h \le r$ (and similarly for $r - h$). 

To bound the last term in \eqref{eq:eigen-diff-bound}, recall that $p_1 + p_2 < \frac{3}{4} < e^{-1/4}$ and $e^\tau \ge \tau^2 / 2$ for all reals $\tau \ge 0$. We thus have
\[
  r(r - 1)(p_1 + p_2)^{r - 1}
  \le r(r - 1)\!\left(\frac{3}{4}\right)^{r - 1}
  = \frac{4}{3}\, r(r - 1)\!\left(\frac{3}{4}\right)^r
  < \frac{4}{3}\, r^2 e^{-r/4}
  \le \frac{4}{3} \cdot 32
  < 50.
\]

From this and \eqref{eq:eigen-diff-bound}, we have $|\lambda_r(m, q, t)| < |p_1 - p_2|^r + \frac{50}{q} \le \delta + \frac{50 n}{m}$ for all $1 \le r \le q$. Hence, $\lambda(J(m, q, t)) = \max_{1 \le r \le q} |\lambda_r(m, q, t)| < \delta + \frac{50 n}{m}$, and Proposition~\ref{prop:eml} gives
\begin{align*}
  \Prb_{(S, S') \sim E(J(m, q, t))}\!\left[S \in \cT \;\land\; S' \in \cT\right]
  &\le \alpha^2 + \lambda\big(J(m, q, t)\big)\,\alpha(1 - \alpha) \le \alpha^2 + \alpha\!\left(\delta + \frac{50 n}{m}\right). \qedhere
\end{align*}
\end{proof}

\section{Conclusion}\label{sec:conclusion}

In this paper, we resolved the open question of Benad\`e et al.~\cite{BHPV24} on the existence of envy-free allocations for arbitrary non-additive valuation functions in the random renaming model. We proved that for any number of agents $n \ge 3$ and any number of items $m$ divisible by $n$, an envy-free allocation exists with high probability whenever $m = \omega(n^{360})$. Our proof requires no monotonicity, subadditivity, or positivity assumptions on the base valuations $v_1, \dots, v_n$, and relies on a combination of the second moment method, top-$\alpha$ subset ranking, and spectral analysis of the generalized Johnson graph $J(m, q, t)$~\cite{BCIM18,KR25}.

An interesting open question is to close the polynomial gap between our result and the birthday-paradox lower bound $m = \Omega(n^2)$ from~\cite{BHPV24}. In particular, does $m = \omega(n^2)$ (with $n \mid m$) suffice for an envy-free allocation to exist with high probability for arbitrary valuation functions, as it does for order-consistent valuations? Reaching the $m = \omega(n^2)$ threshold appears to encounter a natural barrier in our spectral second-moment approach: our spectral second-moment bound requires $\alpha^2 + \alpha \max_{r \ge 1} |\lambda_r(m, q, t_i)| \le \alpha^2(1 + o(1/n))$, i.e., $\max_{r \ge 1} |\lambda_r(m, q, t_i)| = o(\alpha / n) \ll 1/n^3$ (as $\alpha = o(1/n^2)$ is needed for the first-moment union bound). However, for every integer $t_i \in \{0, \dots, q - 1\}$, \eqref{eq:eigenvalue-formula} yields $\max_{r \ge 1} |\lambda_r(m, q, t_i)| \ge \Omega(1/m)$, which precludes $\max_{r \ge 1} |\lambda_r(m, q, t_i)| \ll 1/n^3$ whenever $m = O(n^3)$. Thus, improving the threshold to $m = \omega(n^2)$ for general valuations will likely require substantially new ideas beyond black-box spectral bounds on $J(m, q, t)$.

\section*{Acknowledgments}

We thank Warut Suksompong and Daniel Halpern for helpful discussions.

\section*{AI Disclosure}

The top-$\alpha$ subset ranking idea and its connection to the expansion of generalized Johnson graphs were discovered by the author, whereas an AI assistant (Gemini) discovered the proofs of the expansion bounds (in Section~\ref{sec:fourier-proofs}) and helped draft the paper. The author has verified and heavily edited all the proofs, and takes full responsibility of the correctness of the results.

\bibliographystyle{alpha}
\bibliography{references}

\end{document}